\documentclass{amsart}

\usepackage{tikz}
\usetikzlibrary{patterns}
\usepackage{caption}

\usepackage{amsmath, amsthm, amssymb, delarray, todonotes, graphicx, subcaption, float, url, hyperref, tikz-cd} 
\usepackage[margin=1in]{geometry} 
\usepackage{enumerate}
\usepackage{tikz} 
\usetikzlibrary{cd}
\usepackage{comment}

\usepackage[all]{xy}          
\xyoption{dvips}

\makeatletter
\renewcommand\l@subsection{\@tocline{2}{0pt}{2pc}{5pc}{}}
\makeatother

\newtheorem*{rep@theorem}{\rep@title}
\newcommand{\newreptheorem}[2]{%
	\newenvironment{rep#1}[1]{%
		\def\rep@title{#2 \ref{##1}}%
		\begin{rep@theorem}}%
		{\end{rep@theorem}}} 

\theoremstyle{plain}
\newtheorem{thm}{Theorem}[section]
\newreptheorem{thm}{Theorem}

\newreptheorem{thmp}{Theorem} 
\newtheorem{prop}[thm]{Proposition}

\theoremstyle{definition}
\newtheorem{defin}[thm]{Definition}
\newtheorem{example}[thm]{Example}
\newtheorem{def/ex}[thm]{Definition/Example}

\theoremstyle{remark}
\newtheorem{rem}[thm]{Remark}

\newcommand{\refS}[1]{Section~\ref{S:#1}}
\newcommand{\refT}[1]{Theorem~\ref{T:#1}}

\newcommand{\refP}[1]{Proposition~\ref{P:#1}}
\newcommand{\refD}[1]{Definition~\ref{D:#1}}

\newcommand{\R}{{\mathbb R}}
\newcommand{\Z}{{\mathbb Z}}

\begin{document}


\title[Topology of simple games]{The topology of simple games}


\author{Leah Valentiner}
\address{Department of Mathematics, Wellesley College, 106 Central Street, Wellesley, MA 02481}
\email{lv101@wellesley.edu}

\author{Ismar Voli\'c}
\address{Department of Mathematics, Wellesley College, 106 Central Street, Wellesley, MA 02481}
\email{ivolic@wellesley.edu}
\urladdr{ivolic.wellesley.edu}



\begin{abstract}
 We initiate the study of simple games from the point of view of combinatorial topology. The starting premise is that the losing coalitions of a simple game can be identified with a simplicial complex. Various topological constructions and results from the theory of simplicial complexes then carry over to the setting of simple games. Examples are cone, join, and the Alexander dual, each of which have interpretations as familiar game-theoretic  objects. We also provide some new topological results about simple games, most notably in applications of homology of simplicial complexes to weighted simple games. The exposition is introductory and largely self-contained, intended to inspire further work and point to what appears to be a wealth of potentially fruitful directions of investigation bridging game theory and topology.
\end{abstract}

\keywords{Simple games, weighted games, hypergraphs, simplicial complexes, Alexander duality.}


\maketitle 

\tableofcontents


\parskip=6pt
\parindent=0cm


\section{Introduction}\label{S:Intro}


The goal of this paper is to initiate a study of simple games using topology. The bridge between the two is a simple observation that losing coalitions in a simple game form a simplicial complex. More precisely, a simple game on a set of players $N$ declares each subset of $N$ to be a winning or a losing coalition with the condition that supersets of winning coalitions are also winning. This condition can also be restated as a requirement that subsets of losing coalitions are also losing, and this is where the connection to topology can be made. Namely, if we think of each losing coalition of size $k$ as a $(k-1)$-simplex in a Euclidean space (so a losing coalition of, say, three voters is represented as a 2-simplex, or a triangle), then the subset condition says that each face of a simplex also must represent a losing coalitions. Putting these together produces a simplicial complex, a standard object in topology -- a space that is defined combinatorially and thus easier to work with than a general topological space.

Since losing coalitions determine a game, there is a correspondence between simple games on $n$ players and simplicial complexes on $n$ or fewer vertices. In other words, simplicial complexes contain all the information about simple games. It is therefore not unexpected that topological properties and results about simplicial complexes might translate into information about simple games. As we will see, this is indeed the case. Topology can recover much of what is already known from simple game theory, but can also supply new information. For example, standard notions of strong and proper games have an interpretation in the simplicial complex world, as do those of dummy player, dictator, and other special kinds of players. Sums and products of games and dual games also have counterparts in topology.  

The appearance of simplicial complexes in game theory is not unexpected. It is a familiar (although underutilized) point of view that the winning coalitions can be represented by a hypergraph, which is a generalization of a simplicial complex (and a graph) that does not require the faces of all simplices to be present. This makes sense since subsets of winning coalitions need not be winning, only supersets. As far as we know, the hypergraph approach was never explored in depth (even though the canonical reference on the subject, Taylor and Zwicker's book \emph{Simple Games} \cite{TZ:SimpleGames}, defines a simple game as a hypergraph on the first page). This might simply be because hypergraphs are topologically more complicated than simplicial complexes. However, adopting the complementary view that simple games are determined by their \emph{losing} coalitions pivots the theory away from hypergraphs and to simplicial complexes, with the full richness of the topology that comes along with them.

Another reason not to be surprised by the role of simplicial complexes in game theory is that topology has in fact been present in it for a while. Various fixed point theorems, which have origins in topology, are used to prove game-theoretic results, including Nash Equilibrium \cite{C:FixedPoints, T:BrowerNash, Y:FixedPoints}. In the 1980s, Chichilnisky initiated a topological analysis of social choice theory \cite{C:SocialChoiceTopology, C:ParetoDictator, C:SocialChoiceTopologyRecent} which led to Baryshnikov's topological proof of Arrow's Theorem in the 1990s \cite{B:Impossibility, B:TopDiscSocChoice, B:AroundArrow} and more recent work on the topological versions of the Gibbard-Satterthwaite Theorem \cite{BR:GSThm}. A more recent combinatoral topology proof of Arrow's Theorem and its generalization can be found in \cite{LRR-P:GeneralArrow, RR-P:Arrow}. Related recent work on modeling political strucures with simplicial complexes \cite{AK:Conflict, BSV:SimplicialPower, MV:PoliticsComplex, VW:PoliticsHypergraph} also demonstrates the power of the simplicial complex approach.

In Sections \ref{S:SimpleGames} and \ref{S:SC}, we give some basic background on simple games and simplicial complexes, respectively. Most of this material is standard and we provide ample references for further reading. One exception is \refP{AlexanderJoinIntersection}, relating Alexander duals and joins, for which we could not find a proof in the literature and which might be of independent interest. In Section \ref{S:Correspondence}, we bring the two together and establish the dictionary that translates between the two fields. Most of the statements there are interpretations of familiar simple games results, with the exception of \refT{HomologySymmetric}, which gives a homological characterization of symetric weighted games. We hope this result points the way toward more sophisticated applications of topology, and in particular some standard topological invariants like homology, in game theory. The paper closes with some possible further avenues of investigation in Section \ref{S:FurtherWork}.


\subsection{Acknowledgments} Ismar Voli\'c is grateful to the Simons Foundation for its support.



\section{Simple Games}\label{S:SimpleGames}


The notion of a simple (voting) game goes back to Shapley \cite{S:SimpleGames} (with a narrower definition appearing even earlier \cite{vNM:GameTheory}). This is a basic class of examples of cooperative games that packs a fair amount of complexity despite its straightforward definition. A subclass of simple games is that of weighted voting games which can be found in many real-life applications and institutions.

Here we provide some basic background on simple games. For details, see, for example,  \cite{FM:VotingPower, TZ:SimpleGames}.

\subsection{Simple games and coalitions}\label{S:Coalitions}


\begin{defin}\label{D:SimpleGame}
A \emph{simple game} $G$ is a pair $(N,v)$ where $N=\{p_1, p_2, ..., p_n\}$ is a finite set  and $v$ is a function 
$$v\colon \mathcal{P}(N) \longrightarrow \{0, 1\}$$
satisfying \emph{monotonicity}: 
\begin{center}
 if $v(S) = 1$ and $S \subseteq T \subseteq N$, then $v(T) = 1$. 
 \end{center}

Here $\mathcal{P}(N)$ denotes the power set of $N$, the set of all its subsets. Elements $p_i$ of $N$ are called \emph{players} (or \emph{voters} or \emph{agents}). A subset $S$ of $N$ is a \emph{coalition}. The set $N$ is sometimes called the \emph{grand coalition}. If $v(S)=1$, then the coalition is \emph{winning}. Otherwise it is \emph{losing}. 
\end{defin}

This definition models the setting where a set of players is making a binary choice that can be thought of as an alternative being presented against the status quo, like a legislative body deciding to adopt a resolution or not. Another setting is that of elections where players are voters who are expressing their preference on one of two candidates. 

Assignment $v$ is called the \emph{characteristic} or \emph{worth function}. If  $v(S)=1$, then $S$ is a collection of players that is able to enact the alternative. Monotonicity should make sense intuitively since it simply states that if a winning coalition gets bigger, it is still winning.

An alternative to the ``upward'' monotonicity requirement above is to specify that losing coalitions satisfy ``downward monotonicity,'' i.e.
\begin{equation}\label{E:DownMonotonicity}
\text{if } v(S) = 0 \text{ and } N\supseteq S \supseteq T, \text{ then } v(T) = 0.
\end{equation}
It is this characterization that will lead to the identification of simple games with simplicial complexes later.
\footnote{Sometimes in the literature, the definition of a simple game specifies that the grand coalition $N$ is winning, i.e.~$v(N)=1$ and that the empty coalition is losing, so $v(\emptyset)=0$. Note that, because of monotonicity, the only way the grand coalition could be losing is if all coalitions are losing and the empty coalition could only win if all coalitions are winning.}

A simple game can thus also be described a set $\mathcal{W}\subset \mathcal{P}(N)$ of winning coalitions since that determines the set of losing ones. Conversely, one could specify the losing coalitions $\mathcal{L}\subset \mathcal{P}(N)$. 

\begin{defin}\label{D:Coalitions}\ 
\begin{itemize}
\item A winning coalition is \emph{minimal} if all its proper subsets are losing coalitions. 
\item A losing coalition is \emph{maximal} if all its proper supersets are winning.
\end{itemize}
\end{defin}

 Because of the monotonicity requirement, it therefore suffices to specify the set of minimal or maximal coalitions  to describe the game.
 
 \begin{example}
 Let $N=\{p_1,p_2,p_3,p_4\}$. Let the minimal winning coalitions be given by
 $$
 v(\{p_1,p_2\})=v(\{p_2,p_3,p_4\})=1.
 $$
 This determines the game since we know that $\{p_1,p_2,p_3\}$, $\{p_1,p_2,p_4\}$, and $\{p_1,p_2,p_3, p_4\}$ are also winning coalitions (for a total of five), and all other subsets of $N$ are losing coalitions.
 \end{example}

\begin{defin}\label{D:Blocking}A losing coalition $S$ is \emph{blocking} if $N\setminus  S$ is losing.
\end{defin}

\begin{example}\label{Ex:NotConstantSum}
Consider $N=\{p_1, p_2, p_3, p_4\}$ with winning coalitions defined as those that contain at least one of $p_1$ or $p_2$ and 
at least one of $p_3$ or $p_4$. Then, for example, $\{p_1, p_2\}$ is a blocking coalition since it is losing, and its complement, $\{p_3, p_4\}$, is also losing. 
\end{example}


\subsection{Proper and strong games}\label{S:ProperStrong}


In Example \ref{Ex:NotConstantSum}, in addition to there being two disjoint losing coalitions, namely $\{p_1, p_2\}$ and $\{p_3, p_4\}$, there also exist two disjoint winning coalitions, namely $\{p_1, p_3\}$ and $\{p_2, p_4\}$. This motivates the following terminology.

Let $S^c=N\setminus S$ denote the complement of $S$ in $N$.

\begin{defin}\label{D:KindsOfGames}
\ 
\begin{itemize}
\item
A \emph{proper simple game} satisfies the property that if $S$ is a winning coalition, then the complement $S^c$ is not.

\item
A \emph{strong simple game} satisfies the property that if $S$ is a losing coalition, then $S^c$ is not.

\item
A simple game is \emph{constant sum} if it is both proper and strong.

\end{itemize}

\end{defin}

Thus a game is proper if there cannot be two disjoint winning coalitions whose union is $N$ and it is strong if the same holds for losing coalitions.

Any game where the only winning coalition is $N$ itself (unanimity) is not strong. If the only losing coalition is empty, then the game is not proper. The U.S.~jury system can be regarded as either, depending on whether the value of 1 corresponds to the defendant being innocent or guilty.


If a game is constant sum, it follows that each partition of $N$ into two subsets must include one winning coalition and one losing coalition.

\begin{defin}\label{D:KindsOfPlayers}
Suppose $G=(N,v)$ is a simple game with winning coalitions $\mathcal W$.
\begin{itemize}
\item Player $p_i$ is a \emph{dummy} if, for all $S\in \mathcal{W}$, $S\setminus \{p_i\}\in \mathcal{W}$.
\item Player $p_i$ is a \emph{vetoer} if, whenever $p_i\notin S$,  $S\notin \mathcal{W}$.
\item Player $p_i$ is a \emph{passer} if,  whenever $p_i\in S$,  $S\in \mathcal{W}$.
\item Player $p_i$ is a \emph{dictator} if $S\in \mathcal{W}$ if and only if $\{p_i\}\in S$. 
\end{itemize}
\end{defin}

Note that a player who is a vetoer and a passer is in fact a dictator.

If $p_i$ is a vetoer, then $\{p_i\}$ is a blocking coalition. For example, in the process of adopting a United Nations Security Council resolution,  any of the five permanent members is a vetoer.


\subsection{Arithmetic of games}\label{S:Algebra}


There are two natural ways to create new games out of existing ones. 

\begin{defin}\label{D:SumProduct}
Suppose $G=(N,v)$ and $G'=(N',v')$ are simple games. 
\begin{itemize}
\item The \emph{sum} $G+G'$ of $G$ and $G'$ is defined on players $N\cup N'$ by
$$
S\in N\cup N' \text{ is winning in } G+G' \Longleftrightarrow S\cap N \text{ is winning in } G \text{ or } S\cap N' \text{ is winning in } G'.
$$
\item The \emph{product} $G\times G'$ of $G$ and $G'$ is defined on players $N\cup N'$ by
$$
S\in N\cup N' \text{ is winning in } G\times G' \Longleftrightarrow S\cap N \text{ is winning in } G \text{ and } S\cap N' \text{ is winning in } G'.
$$
\end{itemize}

\end{defin}

Alternatively, the sum and product can be defined in terms of losing coalitions by stating that $S$ is a losing coalition for the sum if it is losing in both $G$ and $G'$ (after intersecting with $N$ and $N'$). Similarly for the product, $S$ is losing if it is losing in at least one of the games.

We will denote by $\mathcal L + \mathcal L'$ and $\mathcal L \times \mathcal L'$ the losing coalitions of $G+G'$ and $G\times G'$, respectively.

If we think of $N$ and $N'$ as disjoint sets of players (which they do not have to be in \refD{SumProduct}), then the sum models a parallel decision-making situation when two groups are independently deciding on an issue and either can enact a win for the larger game. For the product, a win has to occur in both games individually. One can think of two chambers of legislature and the situation where either can push legislature through or both have to approve it for passage.


\subsection{Dual games}\label{S:DualGame}


\begin{defin}\label{D:DualGame}
Given a simple game $G=(N,v)$, the \emph{dual game} $G^*=(N,v^*)$ is defined as 
$$
v^*(S)=1 \Longleftrightarrow v(S^c)=0.
$$
\end{defin}
In other words, $S$ is a winning coalition in the dual game if and only if its complement is a losing coalition in the game. In light of \refD{Blocking}, winning coalitions of $G^*$ are precisely the blocking coalitions of  $G$.

It is not hard to see that $(G^*)^*=G$. In addition, unravelling the definitions yields the following standard result (see, for example, \cite[Proposition 1.3.7]{TZ:SimpleGames}).

\begin{prop}\label{P:SelfDualGame}
A simple game $G$ is constant sum if and only if it is self-dual, namely if $G=G^*$.
\end{prop}

Just as one can denote the game in terms of winning or losing coalitions $\mathcal W$ and $\mathcal L$, we will sometimes label the dual game in terms of its winning or losing coalitions as $\mathcal W^*$ or $\mathcal L^*$.

\begin{example}\label{Ex:SelfdualGame}
Consider the game with players $N=\{p_1, p_2, p_3, p_4, p_5, p_6\}$ and minimal winning coalitions: $\{p_1, p_2\}$; $p_1$ along with any three of $\{p_3, p_4, p_5, p_6\}$; and $p_2$ along with any three of $\{p_3, p_4, p_5, p_6\}$. In each of these cases, the remaining players do not have enough votes to form a winning coalition. It is then easy to check that this game is self-dual.
\end{example}

One general feature of duality in different contexts is that it flips sums and products. The same is true in simple games. For the proof of the following, see for example \cite[Proposition 1.4.3.]{TZ:SimpleGames}. We will later provide  results with which this proposition can be recovered topologically.

\begin{prop}\label{P:DualitySumProduct}\ 
\begin{itemize}
\item $(G+G')^* = G^*\times G'^*$
\item $(G\times G')^* = G^*+ G'^*$
\end{itemize}

\end{prop}


\subsection{Maps between games}\label{S:MapsGames}


Suppose $(N,v)$ and $(N',v')$ are simple games. A function 
$$
f\colon N\longrightarrow N'
$$
determines the function 
$$
\phi\colon \mathcal P(N) \longrightarrow \mathcal P(N')
$$
in a clear way.
\begin{defin}\label{D:Homomorphism}
A \emph{map} or a \emph{homomorphism} of games $G=(N,v)$ and $G=(N',v')$ is a function $\phi$ as above such that 
$$
S\subset N \text{ winning in $G$}\ \Longrightarrow\ \phi(S)\subset N' \text{ winning in $G'$}.
$$
If $f$ is a bijection and the diagram
\begin{equation}
	\begin{tikzcd}[row sep=small]
		\mathcal P(N) \arrow{dd}{\phi}\arrow{dr}{v}  & \\
		 & \{0,1\} \\
		\mathcal P(N')\arrow{ur}{v'}& 		
	\end{tikzcd}
\end{equation}
commutes, then $\phi$ is an \emph{isomorphism}.
\end{defin}
A homomorphism preserves the winning structure of a game, but not necessarily the losing structure since losing coalitions could get sent to winning coalitions. If both are preserved and $\phi$ is induced by a bijection, then we have an isomorphism since the above diagram says precisely that a coalition $S$ in $G$ is winning/losing if and only if the coalition $\phi(S)$ is winning/losing in $G'$. In other words, there is a correspondence between the players and the coalitions of the two games,  so the games have exactly the same structure.

A homomorphism can preserve both winning and losing coalitions without being an isomorphism. This is easy to see by taking any isomorphism and adding a dummy player to the target game. Then $f$ is not a bijection (since it is not a surjection), but winning/losing coalitions still get sent to winning/losing coalitions. 

Suppose $(N,v)$ is a simple game and $f\colon N\to N'$ is surjective. We can build a new simple game from $N'$ by declaring that 
$$
S\in \mathcal P(N') \text{ is winning/losing } \Longleftrightarrow
f^{-1}(S)\in \mathcal P(N) \text{ is winning/losing } 
$$ 
In particular, for each player $p'_i\in N'$, one can look at the fiber $f^{-1}(p'_i)$ which may consist of one or more players in $N$. This models the situation when the players $f^{-1}(p'_i)$ always act in unison, namely they always align in some way in a game (e.g.~vote the same way in elections). Such set of players is called a \emph{block}. The construction described above can be thought of as passing from the original game to a new one where blocks of voters are regarded as a single voter. 

\begin{rem}
Maps between games that are induced from surjective functions on players can be used to define the \emph{Rudin-Keisler ordering} on simple games \cite[p.24]{TZ:SimpleGames}.
\end{rem} 




%


\subsection{Weighted games}


Weighted games are a rich source of examples of simple games. They model the situation where players carry different weights in the decision-making process, such as the U.S.~Electoral College (where states are regarded as players with weights corresponding to the numbers of electoral votes they have), United Nations Security Council, European Parliament, International Monetary Fund, or legislative bodies of multi-party systems where representatives of each party vote as a block.

\begin{defin}
A \emph{weighted} game consists of a finite set $N=\{p_1, p_2, ..., p_n\}$ of players, a quota $q \in \mathbb{R}$, and a \emph{weight function} 
$$w\colon N \longrightarrow \mathbb{R_+}$$ that assigns a \textit{weight} to each player. The winning coalitions are those for which the sum of the weights of the coalition members meets or exceeds $q$, namely
$$
v(S)=1\ \  \Longleftrightarrow\ \  \sum_{p_i\in S}w(p_i)\geq q.
$$
\end{defin}

Setting $w_i=w(p_i)$, a common notation for a weighted game is 
$$
V(q; w_1, w_2, ..., w_n).
$$

A weighted game is \emph{symmetric} if each player has weight 1. In this case, a coalition is winning if it contains at least $q$ players. A symmetric game has the form $V(q;1,1,...,1)$, where 1 is repeated $n$ times, but it is not always initially clear that a weighted game is symmetric. For example, $V(5;4,3,2)$ has the same winning coalition structure as $V(2;1,1,1)$ and is therefore isomorphic to it according to \refD{Homomorphism}. Establishing criteria for determining whether a weighted game is isomorphic to a simple one is therefore useful. One such topological criterion is given in \refT{HomologySymmetric}.

If $t$ is the total sum of weights of all the players, the quota $q$ is typically an integer between $t/2$  (simple majority) and $t$ (unanimity). If this is the case, then the game is proper since there cannot be two disjoint winning coalitions, each adding up to the number of votes that is greater than $t/2$. However, it is not true that such a game is also automatically strong; take, for example, $V(3; 1,1,1,1)$, where the first two players and the last two players form losing coalitions that are complements of each other.

\begin{example}
The game from Example \ref{Ex:SelfdualGame} is in fact a weighted game: Give players $p_1$ and $p_2$ weight 4, all the other players weight 1, and set the quota at 7.
\end{example}

Not every simple game is a weighted game, as illustrated in the following example. 

\begin{example}If $N=\{p_1, p_2, p_3, p_4\}$ and the winning coalitions are $\{p_1, p_2\}$ and $\{p_3, p_4\}$, then if we try to represent this as a weighted game, we must necessarily have $w(p_1)+w(p_2)\geq q$ and  $w(p_3)+w(p_4)\geq q$, so $\sum w(p_i)\geq 2q$. On the other hand, $w(p_1)+w(p_3)< q$ and $w(p_2)+w(p_4)< q$ (since $\{p_1, p_3\}$ and $\{p_2, p_4\}$ are losing coalitions), so $\sum w(p_i)< 2q$, which is a contradiction.
\end{example}

A real-life example of a simple game that is not weighted is the mechanism to amend the Canadian Constitution \cite[p.10-11]{TZ:SimpleGames}. An important endeavor in the study of simple games is to find criteria for determining when a game is weighted. Some of those include \emph{trade-robustness} \cite{TZ:WeightedVoting} and \emph{acyclicity} \cite[Theorem 5.4.6]{TZ:SimpleGames}.

\section{Simplicial complexes}\label{S:SC}


We now turn our attention to simplical complexes, objects belonging to the field of combinatoral topology since they are defined set-theoretically but can also be regarded as topological spaces. This makes them extremely useful as they are easier to handle than generic spaces. For example, computing homology for simplicial complexes is fairly straigthtforward while this is far from true in general. Realizing a space as topologically same as a simplicial complex (homeomorphic or homotopy equivalent) thus makes it a lot easier to understand.

Here we provide just the basic definitions; details can be found in, for example \cite{FPR09, Mat03, MS:CommAlg, Spa95}.


\subsection{Combinatorial and geometric simplicial complexes}\label{S:CombGeomSC}


For $S$ a set, let $\mathcal P(S)$ denote its power set as usual.

\begin{defin}\label{D:SC}
A (\emph{combinatorial} or \emph{abstract}) \emph{simplicial complex} is the pair $K=(V,\Delta)$ that consists of a finite set $V=\{v_1, v_2, ..., v_n\}$ called \emph{vertices} and a set $\Delta\subset\mathcal P(V)$ called \emph{simplices} satisfying
\begin{enumerate}
\item $\{v_i\}\in\Delta$ for all $1\leq i\leq n$;
\item if $\sigma \in \Delta$ and $\tau \subset \sigma$, then $\tau \in \Delta$.
\end{enumerate}
\end{defin}

Since $\Delta$ already contains information about $V$ (because of condition (1)),  $K$ can be identified with $\Delta$.

\begin{example}\label{E:complex1}
An example of a simplicial complex on the vertex set $V=\{v_1, v_2, v_3, v_4\}$ is 
$$
K=\{\{v_1\}, \{v_2\}, \{v_3\}, \{v_4\}, \{v_1, v_2\}, \{v_1, v_3\}, \{v_2, v_3\}, \{v_2, v_4\}, \{v_1, v_2, v_3\} \}.
$$
\end{example}

\begin{example}\label{E:complex2}
An example of a simplicial complex on the vertex set $V=\{v_1, v_2, v_3, v_4, v_5, v_6\}$ is 
\begin{align*}
K= & \{\{v_1\}, \{v_2\}, \{v_3\}, \{v_4\}, \{v_5\}, \{v_6\}, \{v_1, v_3\}, \{v_1, v_4\}, \{v_1, v_5\}, \{v_1, v_6\}, \{v_2, v_3\}, \{v_2, v_4\}, \{v_2, v_5\}, \{v_2, v_6\} \\
& \{v_3, v_4\}, \{v_3, v_5\}, \{v_3, v_6\}, \{v_4, v_5\}, \{v_4, v_6\}, \{v_5, v_6\} \\
& \{v_1, v_3, v_4\}, \{v_1, v_3, v_5\}, \{v_1, v_3, v_6\}, \{v_1, v_4, v_5\}, \{v_1, v_4, v_6\}, \{v_1, v_5, v_6\}, \{v_2, v_3, v_4\}, \{v_2, v_3, v_5\}, \{v_2, v_3, v_6\}, \\
&\{v_2, v_4, v_5\}, \{v_2, v_4, v_6\}, \{v_2, v_5, v_6\}, \{v_3, v_4, v_5, v_6\}\}.
\end{align*}

On the other hand, the collection 
$$\{\{v_1\}, \{v_2\}, \{v_3\}, \{v_4\}, \{v_5\}, \{v_6\}, \{v_1, v_2\}, \{v_1, v_4\}, \{v_1, v_2, v_4\}\}$$
is not a simplicial complex on $V$ because $\{v_2,v_4\}$ is not in it.
\end{example}

If $K$ contains all possible nonempy subsets of $V$, then $K$ is the \emph{standard $(n-1)$-simplex} and is denoted by $\Delta^{n-1}$. In other words,
$$
\Delta^{n-1} = \mathcal P_0(V)
$$
where the right side is the power set of $V$ excluding the empty set. One useful way of thinking about how to construct a simplicial complex is  by starting with $\Delta^{n-1}$ and then removing subsets from it. In order to satisfy condition (2) from  \refD{SC}, if a subset is removed, then all its supersets also must be removed.

\begin{example}
A simplicial complex whose simplices are all of cardinality two is precisely a graph.
\end{example}

To turn a (combinatorial) simplicial complex $K$ into a topological space, one forms its \emph{geometric realization} by associating to each vertex a point, to each two-element subset a line segment, to each three-element subset a triangle, to each four-element subset a tetrahedron, etc. The edges, triangles, and their generalizations thus correspond to simplices, and we refer to them as \emph{geometric} simplices, although the distinction between a set and its geometric counterpart will be blurred very soon. All this takes place in a Euclidean space of large enough dimension and the realization is topologized as a subset thereof. The geometric simplices are non-degenerate in the sense that the vectors $v_i-v_1$ are independent (now $v_i$ are thought of as points in a Euclidean space, so thinking of $v_i-v_1$ as vectors makes sense). 

The geometric realization of $K$ is denoted by $|K|$, but we will refer to both simply as $K$. This will not cause confusion.

The simplices fit together along common subsets which are called \emph{faces}. If a face is determined by $m$ vertices (i.e.~the face is a convex hull of $m$ vertices, so if $m=3$ the face is a triangle, if $m=4$ the face is a tetrahedron etc.), then it is called an \emph{$(m-1)$-simplex} or a simplex of \emph{dimension $m-1$}. In the notation following \refD{SC}, such a face can also be denoted by $\Delta^{m-1}$.

\begin{example}
The geometric realizations of the simplicial complexes from Examples \ref{E:complex1} and \ref{E:complex2} are given in Figures \ref{F:SimpComplex1} and \ref{F:SimpComplex2}.
\end{example}

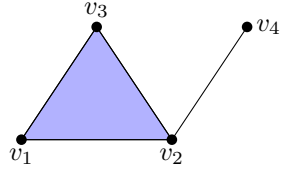
\begin{figure}
\begin{center}
\begin{tikzpicture}
    \coordinate (v1) at (0,0);
    \coordinate (v2) at (2,0);
    \coordinate (v3) at (1,1.5);
    \coordinate (v4) at (3,1.5);

    \fill[blue!30] (v1) -- (v2) -- (v3) -- cycle;
    \draw (v1) -- (v2) -- (v3) -- cycle;
    
    \draw (v1) -- (v3);
    \draw (v2) -- (v3);
    \draw (v1) -- (v2);

    \draw (v2) -- (v4);

    \foreach \point/\label/\anchor in {v1/$v_1$/north, v2/$v_2$/north, v3/$v_3$/south, v4/$v_4$/west} {
        \fill[black] (\point) circle (2pt);
        \node[anchor=\anchor] at (\point) {\label};
    }
\end{tikzpicture}
\end{center}
\caption{The geometric realization of the simplicial complex from Example \ref{E:complex1}.}
\label{F:SimpComplex1}
\end{figure}

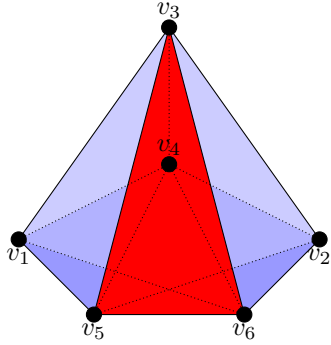
\begin{figure}
\begin{center}
\begin{tikzpicture}
\begin{scope}
\fill[blue!20] (0,0) -- (2,1) -- (2, 2.82);
\fill[blue!20] (4,0) -- (2,1) -- (2, 2.82);
\fill[blue!30] (4,0) -- (3,-1) -- (2, 1);
\fill[blue!40] (4, 0) -- (1, -1) -- (3, -1);
\fill[blue!30] (0, 0) -- (3, -1) -- (2, 1);
\fill[blue!40] (0, 0) -- (1, -1) -- (3, -1);
\fill[red] (1, -1) -- (3, -1) -- (2, 2.82);

\end{scope}
\draw[fill] (0,0) circle [radius=0.1]; \draw[fill] (2, 1) circle [radius=0.1]; \draw[fill] (4,0) circle [radius=0.1]; \draw[fill] (2,2.82) circle [radius=0.1]; \draw[fill] (1, -1) circle [radius=0.1]; \draw[fill] (3, -1) circle [radius=0.1];
\draw (2,2.82) node[above] {$v_3$} -- (0,0) node[below] {$v_1$}; 
\draw[densely dotted] (0,0) -- (2, 1) node[above] {$v_4$} -- (2, 2.82);
\draw (2, 2.82) -- (4, 0) node[below] {$v_2$};
\draw[densely dotted] (4, 0) -- (2, 1);
\draw[densely dotted] (2, 1) -- (1, -1) node[below] {$v_5$};
\draw (0, 0) -- (1, -1);
\draw[densely dotted] (2, 1) -- (3, -1) node[below] {$v_6$};
\draw[densely dotted] (3, -1) -- (0, 0);
\draw (2, 2.82) -- (1, -1);
\draw (2, 2.82) -- (3, -1) -- (4, 0);
\draw[densely dotted] (4, 0) -- (1, -1);
\draw (1, -1) -- (3, -1);

\end{tikzpicture}
\end{center}
\caption{The geometric realization of the simplicial complex from Example \ref{E:complex2}. The red tetrahedron is solid. The blue faces are 2-simplices and only intersect at common vertices and edges.}
\label{F:SimpComplex2}
\end{figure}

Here is additional useful terminology.

\begin{defin}\ 

\begin{itemize}
\item A \emph{subcomplex} of $K=(V,\Delta)$ is a simplicial complex whose vertices are a subset of $V$ and whose simplices are a subset of $\Delta$. 
\item A simplex is called \emph{maximal} if it is not a face of any other simplex. 
\item The \emph{$i$-skeleton} $K^i$ of $K$ is the collection of all $i$-simplices of $K$ along with their faces. Since $K^i \subset K^j$ for $i \le j$, there is a filtration of $K$ by its skeleta.
\item The \emph{star} of a vertex $v$ is the set $$\text{st}(v)=\{\sigma \in \Delta: v\in \sigma\}.$$
\end{itemize}
\end{defin}


%
%

There are certain standard topological constructions one can perform on simplicial complexes (or on topological spaces more generally) which we will find useful.


\subsection{Constructions}\label{S:Construct}


There are many ways to construct simplicial complexes from others. For example, the intersection of two complexes is a complex. Here is another construction we will need.

\begin{defin}\label{D:Join}
For simplicial complexes $K$ and $K'$ with disjoint vertex sets $V$ and $V'$, their \emph{join} $K\star K'$ is defined to be the simplicial complex with vertex set $V\cup V'$ and simplices
$$
K\cup K'\cup \{\sigma\cup\tau \colon \sigma\in K, \tau\in K'\}.
$$
\end{defin}

A special case is when $K'$ is a single vertex $c$, in which case we get the \emph{cone} on $K$, denoted by $CK$. Its simplices are 
$$
K\cup \{c\} \cup \{\sigma\cup \{c\}\colon \sigma\in K\}.
$$

\begin{example}\label{Ex:Cone}
Suppose $K=\{\{v_1\}, \{v_2\}, \{v_3\}, \{v_4\}, \{v_1, v_2\}, \{v_1, v_3\}, \{v_2, v_3\}, \{v_1, v_4\}, \{v_1, v_2, v_3\}\}$. Then 
\begin{align*}
CK= & \{\{v_1\}, \{v_2\}, \{v_3\}, \{v_4\}, \{v_1, v_2\}, \{v_1, v_3\}, \{v_2, v_3\}, \{v_1, v_4\}, \{v_1, v_2, v_3\}, \\ 
&\{c\}, \{v_1, c\}, \{v_2, c\}, \{v_3, c\}, \{v_4, c\},  \{v_1, v_2, c\}, \{v_1, v_3, c\}, \{v_2, v_3, c\}, \{v_1, v_4, c\}, \{v_1, v_2, v_3, c\} \}.
\end{align*}
\end{example}
Topologically, the join is the space of all line segments joining $K$ and $K'$. In particular, 
 the cone is the quotient $(K\times I)/(K\times \{1\})$, where $I$ is the unit interval. Note that
$$
C\Delta^{n-1}=\Delta^{n}.
$$

Another related construction is that of the \emph{suspension of $K$} which is obtained by gluing two copies of $CK$ along $K$.

Another notion we need is that of the Alexander dual.

\begin{defin}\label{D:AlexDual}
Given a simplicial complex $K$, its \emph{Alexander dual} is the simplicial complex $K^*$ whose vertex set is the same as that of $K$ and whose simplices are
$$
\{\sigma\in \mathcal P(V)\colon \sigma^c \notin K\}.
$$
(Here $\sigma^c=V\setminus\sigma$ as usual denotes the complement.)
\end{defin}

Alexander dual thus consists of the complements of non-faces of $K$. The ``dual'' terminology comes from the fact that $K^*$ is homotopy equivalent to the complement of $K$ in the $(n-2)$-sphere.

\begin{example}\label{Ex:AlexanderDual}
If $K=\{\{v_1\}, \{v_2\}, \{v_3\}, \{v_4\}, \{v_1, v_2\}, \{v_1, v_3\}, \{v_1, v_4\}, \{v_2, v_3\}\}$, then 
$$
K^*=\{\{v_1\}, \{v_2\}, \{v_3\}, \{v_4\}, \{v_1, v_2\}, \{v_1, v_3\}\}.
$$
\end{example}

Here is a result of independent interest relating Alexander duality and joins.

\begin{prop}\label{P:AlexanderJoinIntersection} For $K$ and $K'$ simplicial complexes with vertex sets $N$ and $N'$ such that $N \cap N' = \emptyset$,
\begin{itemize}
\item $(K\star K')^* = (K \star P(N'))^*\cup (K' \star P(N))^*$
\item $K^*\star K'^*=((K \star P(N')) \cup (K' \star P(N)))^*$
\end{itemize}
\end{prop}

\begin{proof}
For the first statement, take $\sigma \in (K\star K')^*$. Then $\sigma^c \notin \{\tau \cup \tau'\colon \tau \in K, \tau' \in K'\}$. So  either $\sigma^c \cap N \notin K$ or $\sigma^c \cap N' \notin K'$. Without loss of generality, assume $\sigma^c \cap N \notin K$. Then $\sigma^c \cap N$ is a non-face of $K$ and thus $\sigma^c = (\sigma^c \cap N) \cup (\sigma^c \cap N')$ is a non-face of $K \star P(N')$. Thus $\sigma \in (K \star P(N'))^*$, so clearly $\sigma \in (K \star P(N'))^* \cup (K' \star P(N))^*$.

For the other containment, suppose $\sigma \in (K \star P(N'))^*\cup (K' \star P(N))^*$. Without loss of generality, take $\sigma \in (K \star P(N'))^*$. Then $\sigma^c$ is a non-face of $K \star P(N')$. Since $\sigma^c \cap N' \in P(N'))$, this implies  $\sigma^c \cap N \notin K$. Therefore $\sigma^c = (\sigma^c \cap N) \cup (\sigma^c \cap N') \notin K \star K'$, so it is a non-face. Thus $\sigma \in K \star K'$.


For second statement, suppose $\sigma \in ((K \star P(N')) \cup (K' \star P(N)))^*$. Then $\sigma^c$ is a non-face of $(K \star P(N')) \cup (K' \star P(N))$, i.e.~$\sigma^c \notin (K \star P(N')) \cup (K' \star P(N))$. Then $\sigma^c \notin K \star P(N')$ and $\sigma^c \notin K' \star P(N)$. Since $\sigma^c \notin K \star P(N')$, we know that $\sigma^c \cap N \notin K$ and since $\sigma^c \notin K' \star P(N)$, we know that $\sigma^c \cap N' \notin K'$. Since the Alexander dual $K^*$ is constructed with respect to the vertex set of $K$, which is $N$, we have $\sigma \cap N \in K^*$. Similarly, we have $\sigma \cap N' \in K'^*$. It follows that $\sigma = (\sigma \cap N) \cup (\sigma \cap N') \in K^*\star K'^*$. 

Now let $\sigma \in K^*\star K'^*$. Then $\sigma \cap N \in K^*$ and $\sigma \cap N' \in K'^*$. Since these duals are taken with respect to $N$ and $N'$ respectively, this implies $\sigma^c \cap N \notin K$ and $\sigma^c \cap N' \notin K'$. Therefore $\sigma = (\sigma^c \cap N) \cup (\sigma^c \cap N') \notin K \star P(N')$ and $\sigma^c = (\sigma^c \cap N) \cup (\sigma^c \cap N') \notin K' \star P(N)$. So $\sigma^c$ is a non-face of $(K \star P(N')) \cup (K' \star P(N))$. And so $\sigma \in ((K \star P(N')) \cup (K' \star P(N)))^*$. This shows the other containment.
%
\end{proof}

One of the main uses of $K^*$ is \emph{Alexander duality}, a seminal result in topology which states that the homology groups of $K$ are the cohomology groups (with a dimension shift) of $K^*$. Since the latter might be easier to calculate, this gives a powerful tool for understanding the topology of simplicial complexes. One application is in low-dimensional topology where Alexander duality allows one to compute the cohomology of knot and link complements in $S^3$ (important for distinguishing knots and links).


\subsection{Simplicial maps}\label{S:SimplMaps}


A way to compare simplicial complexes or carry structures from one to another is done via simplicial maps.

\begin{defin}\label{D:SimplicialMap}
A \emph{simplicial map} $s\colon K\to L$ between simplicial complexes $K$ and $L$ satisfies
\begin{enumerate}
\item $s(V(K))\subset V(L)$;
\item for each $\sigma\in \Delta(K)$, $s(\sigma)\in \Delta(L)$.
\end{enumerate}
\end{defin}

In short, $s$ maps vertices to vertices and simplices to simplices. On the level of realizations, $s$ induces a continuous map $|s|\colon |K|\to |L|$ which is determined by what it does on vertices (the image of each non-vertex point is a linear combination of the images of the vertices).  If $s$ is an isomorphism (it has a simplicial inverse), then $|s|$ is a homeomorphism.

One can regard what we have done so far in the following way: We would like to understand certain objects, namely simplicial complexes. We also have ways of ``getting from one to another,'' namely simplicial maps. This points to a potential utility of formalizing this information into the structure of a \emph{category}. Getting into category theory here would take us too far afield, but it is a well-known and a useful fact that simplicial complexes and simplicial maps form a category $\operatorname{SCom}$ which has many nice properties. It would then be desirable to form a corresponding category of simple games, and we believe that carrying over some of the salient properties of the simplicial complex category to the simple game side could be extremely beneficial. We will make some comments on the category structure of simple games in Section \ref{S:Correspondence}.


\subsection{Homology}\label{S:Homology}


One of the most important invariants in topology (which roughly means quantities that can be associated to topological spaces and remain  unchanged under continuous deformations of spaces) are the \emph{homology groups} of a space. Our brief overview of homology here will be at the level of heuristics; details can be found in \cite[Appendix A.6]{MV:Politics} and most introductory books on algebraic topology. 

For any space $X$, one can define homology groups $H_i(X)$, $i\geq 0$. (We take unreduced homology and work over $\Z$.) Each homology group is a direct sum of copies of $\Z$ and some finite cyclic groups. The number of copies of $\Z$ in $H_i(X)$, i.e.~the rank of the $i$th homology group, is called the \emph{$i$th Betti number} and is denoted by $\beta_i(X)$. 

The zeroth Betti number captures the number of connected components of $X$; if $\beta_0(X)=k$, then $X$ has $k$ distinct connected components. Higher Betti numbers capture the number of holes or voids of various dimensions in the space. So for example, $\beta_1(S^1)=1$ since the circle has one 1-dimensional hole. A more precise way to think about this is that $S^1$ contains a 1-dimensional cycle that is not ``filled-in'' in the space, i.e.~it is not a boundary of anything. For the space consisting of $k$ circles meeting at a common point (a $k$-leafed rose, or a \emph{wedge} of $k$ circles), $\beta_1(k\!-\!\text{rose})=k$.

For the $n$-sphere $S^n$, $\beta_n(S^n)=1$ while all other Betti numbers are zero. For the torus $T^2$, $\beta_0(T^2)=1$, $\beta_1(T^2)=2$, and $\beta_2(T^2)=1$, indicating that torus is connected, has two non-trivial (and non-homologous, i.e.~one cannot be deformed into the other) 1-dimensional cycles (longitude and latitude) and one 2-dimensional hole (going through the middle of the torus).

If a simplicial complex $K$ has a non-trivial homology group in dimension $i$, this means that there is (at least one) nontrivial \emph{$i$-cycle}, meaning a collection of $i$-simplices that fit together to enclose an unfilled ``void.'' For example, if three 1-simplices (edges) form the boundary of an unfilled triangle, that would be recorded as a generator of a copy of $\Z$ in $H_1(K)$. The cycle can be longer than three edges, but as long as it is  not ``filled in'' fully (one way to think about this is that the cycle cannot be contracted to a point inside $K$), this will contribute a copy of $\Z$ to $H_1(K)$, i.e.~increase $\beta_1(K)$ by one. Four 2-simplices (triangles) fitting into a hollow tetrahedron form a non-trivial 2-cycle contributing to $H_2(K)$, and so on.

A more precise calculation of the homology groups of a complex reduces to a linear algebra problem, which is precisely one of the advantages of simplicial complexes. For example, to compute the homology of the $n$-sphere, it's helpful to regard it as homeomorphic to the $(n+1)$-simplex whose interior face is removed (so, for example, in case of $S^2$, this would be a hollow tetrahedron), and then calculate the homology groups of the resulting simplicial complex. The calculation proceeds by writing down the \emph{adjacency matrices} $A_i$ of the simplicial complex, which capture how $i$ and $(i+1)$-simplices fit together, and then calculating the Laplacians $A_i A_i^T + A_{i+1} A_{i+1}^T$. The dimensions of the kernels of these Laplacians are precisely the $\beta_i$'s.


\subsection{Hypergraphs}\label{S:Hypergraphs}


Finally, we mention the notion of a \emph{hypergraph}, which will also have a place in setting up a dictionary between topology and simple games. A hypergraph is a generalization of a simplicial complex obtained by removing condition (2) from \refD{SC}. Thus a hypergraph just consists of some subsets of the set of vertices, but there is no ``downward closure'' condition, i.e.~if a simplex (now called \emph{hyperedge}) is in a hypergraph, its faces do not have to be.

Hypergraphs can also be topologized as subspaces of $\R^n$ and share many of the same properties with simplicial complexes. Most of the definitions and constructions we presented here go through the same way for hypergraphs. One exception is homology, which has to be defined more carefully. One way to compute the homology of a hypergraphs is to look at its \emph{closure} into a simplicial complex (add all the faces needed to satisfy condition (2) of \refD{SC}) and then compute the homology of the resulting complex; call that the homology of the original hypergraph. This is somewhat crude for our purposes (adding faces that were not there initially changes the underlying model captured by the hypergraph because it declares the existence of new coalitions of a certain type that were not there before), so we prefer the notion of \emph{embedded homology} \cite{BLRW, RWZ} which is more nuanced and does not obscure the underlying hypergraph like the closure homology does. More on hypergraphs and embedded homology as it pertains to political structures can be found in \cite{VW:PoliticsHypergraph}.


\section{Correspondence between simple games and simplicial complexes}\label{S:Correspondence}


Recall the definition of a simple game, \refD{SimpleGame}, and the observation following it that it is enough to specify losing coalitions $\mathcal L$ to define the game. Furthermore, because of the monotonicity requirement, losing coalitions have to satisfy condition \eqref{E:DownMonotonicity}. Comparing this with \refD{SC}, we see that this condition is identical to the downward monotonicy requirement for faces of a simplicial complex. More precisely, there is a correspondence

\begin{align*}
\text{players } p_i\ \  &\longleftrightarrow\ \  \text{vertices } v_i \\
\text{losing coalitions} \in \mathcal L\ \  &\longleftrightarrow\ \  \text{simplices} \in \Delta
\end{align*}

With the basic correspondence above, many notions from simple games now have a counterpart in simplicial complexes. For example, a losing coalition is maximal precisely if the corresponding simplex in the simplicial complex is maximal. 

Before we state some more interesting observation, a note on notation: We will continue to refer to the set of losing coalitions as $\mathcal L$ or $\mathcal L(N)$ when we have to remember the underlying set of players. For the corresponding simplicial complex, we will retain $N$ as the set of vertices, but will denote the complex corresponding to $\mathcal L$ by $K_{\mathcal L}$ or $K_{\mathcal L}(N)$.

\begin{prop}\label{P:DummyDictatorSimplex}\ 
\begin{itemize}
\item If the simple game $(N,v)$ has a dummy $d$, then $K_{\mathcal L}(N)$ is the cone on $K_{\mathcal L}(N\setminus d)$.
\item If the simple game $(N,v)$ has a passer $p$, then $K_{\mathcal L}(N)$ does not contain the vertex $p$.
\item If the simple game $(N,v)$ has a dictator, then $K_{\mathcal L}(N)=\Delta^{n-2}$, where the vertices of $\Delta^{n-2}$ are the non-dictator players.
\end{itemize}

\end{prop}

\begin{proof}
If a dummy joins a losing coalition, it remains losing. This means that, for each simplex in $K_{\mathcal L}(N)$, a cone on it with $d$ as the cone point is also a losing coalition. Doing this over all losing coalitions that do not contain the dummy gives the cone as described in the first statement.

If $p$ does not appear in the losing complex $K_{\mathcal L}(N)$, that means that it is not in any losing coalition, which in turn means that it appears in every winning coalition, making it a passer. 

For the last statement, if $(N,v)$ has a dictator, then every coalition without them is losing, which gives the full simplex on $n-1$ vertices, i.e.~$\Delta^{n-2}$. Since the dictator is never in a losing coalition, there are no other simplices in $K_{\mathcal L}(N)$ besides this $\Delta^{n-2}$.
\end{proof}


Recall the notion of the sum and product of simple games from \refD{SumProduct} as well as the join of simplicial complexes from \refD{Join}. The following result gives topological interpretations of the sum and product of games as the join and intersection of the associated simplicial complexes of losing coalitions. 

\begin{prop}\label{P:LosingSimplexSumProduct}
Let $G$ and $G'$ be simple games with losing coalitions $\mathcal L$ and $\mathcal L'$ and suppose $N\cap N'=\emptyset$, so the players of the games are disjoint. Then
\begin{itemize}
\item  $K_{\mathcal L + \mathcal L'}= K_{\mathcal L} \star K_{\mathcal L'}$

\item $K_{\mathcal L \times \mathcal L'} = (K_{\mathcal L} \star \mathcal P(N')) \cup (\mathcal P(N) \star K_{\mathcal L'})$.
\end{itemize}
\end{prop}

\begin{proof}

For the first statement, suppose $\sigma \in K_{\mathcal L + \mathcal L'}$. Then $\sigma$ is losing in $G+ G'$. So $\sigma \cap N$ is losing in $G$ and $\sigma \cap N'$ is losing in $G'$. It follows that $\sigma \cap N \in K_{\mathcal L}$ and $\sigma \cap N' \in K_{\mathcal L'}$. Since there are no other vertices than those from $N$ and $N'$,  $\sigma = (\sigma \cap N) \cup (\sigma \cap N') \in K_{\mathcal L} \star K_{\mathcal L'}$. This means that $K_{\mathcal L + \mathcal L'} \subseteq K_{\mathcal L} \star K_{\mathcal L'}$.

Now suppose $\sigma \in K_{\mathcal L} \star K_{\mathcal L'}$. There are two cases:

Case 1: Let $\sigma$ be contained in one of the sets $N$ or $N'$, say $\sigma = \sigma \cap N$ (and $\sigma \cap N' = \emptyset$). Since $\sigma \in K_{\mathcal L} \star K_{\mathcal L'}$, this implies $\sigma \in K_{\mathcal L}$. Then $\sigma \cap N = \sigma$ loses in $G$. And by definition of a simple game, we know that $\sigma \cap N' = \emptyset$ loses in $G'$. So then $\sigma$ is losing in $G+G'$ and thus $\sigma \in K_{\mathcal L + \mathcal L'}$.

Case 2: Suppose $\sigma \cap N \neq \emptyset$ and $\sigma \cap N' \neq \emptyset$. Since $\sigma \in K_{\mathcal L} \star K_{\mathcal L'}$, we know that $\sigma \notin K_{\mathcal L}$ and $\sigma \notin K_{\mathcal L'}$, so we must have $\sigma = \tau \cup \tau'$ for some $\tau \in  K_{\mathcal L}$ and $\tau' \in  K_{\mathcal L'}$. 
Since $N$ and $N'$ are disjoint, it follows that $\tau = \sigma \cap N$ and $\tau' = \sigma \cap N'$ are the only possible sets such that $\sigma = \tau \cup \tau'$ and $\tau \in K_{\mathcal L}$ and $\tau' \in K_{\mathcal L'}$. We therefore have that $\sigma \cap N = \tau \in K_{\mathcal L}$ and $\sigma \cap N' = \tau' \in K_{\mathcal L'}$. So $\sigma \cap N$ is losing in $G$ and $\sigma \cap N'$ is losing in $G'$. Thus, $\sigma$ is losing in $G+ G'$, so $\sigma \in K_{\mathcal L + \mathcal L'}$. This proves the other containment, $K_{\mathcal L} \star K_{\mathcal L'} \subseteq K_{\mathcal L + \mathcal L'}$.


For the second statement, start with $\sigma \in K_{\mathcal L \times \mathcal L'}$. Then $\sigma \cap N$ is losing in $G$ or $\sigma \cap N'$ is losing in $G'$. Without loss of generality, suppose that $\sigma \cap N$ is losing in $G$. Then $\sigma = (\sigma \cap N) \cup (\sigma \cap N')$ where $\sigma \cap N \in K_{\mathcal L}$ and $\sigma \cap N' \in \mathcal P(N')$. Therefore $\sigma \in K_{\mathcal L} \star \mathcal P(N')$ and so $K_{\mathcal L \times \mathcal L'} \subseteq (K_{\mathcal L} \star \mathcal P(N')) \cup (\mathcal P(N) \star K_{\mathcal L'})$. 

Now suppose, again without loss of generality, that $\sigma \in K_{\mathcal L} \star \mathcal P(N')$. Then since $N$ and $N'$ are disjoint, we know that this is only possible if $\sigma \cap N \in K_{\mathcal L}$ and $\sigma \cap N' \in \mathcal P(N')$. Thus $\sigma \cap N$ loses in $G$ and so $\sigma$ is losing in $G \times G'$. Therefore, $\sigma \in K_{\mathcal L \times \mathcal L'}$. Then $(K_{\mathcal L} \star \mathcal P(N')) \cup (\mathcal P(N) \star K_{\mathcal L'}) \subseteq K_{\mathcal L \times \mathcal L'}$.
%
\end{proof}

It is also evident from comparing Definitions \ref{D:DualGame} and 
\ref{D:AlexDual} that the definition of the dual game corresponds precisely to the Alexander dual, namely
$$
\mathcal L^* = K_{\mathcal L}^*
$$
where the left side as before denotes the losing coalitions of the dual game and the right means the Alexander dual of  $K_{\mathcal L}$.

Using Propositions \ref{P:AlexanderJoinIntersection} and \ref{P:LosingSimplexSumProduct}, one can recover  \refP{DualitySumProduct}; we leave the details to the reader. This provides a topological interpretation of the interaction of duality with sums and products of games.

One familiar result from topology is that if $K$ is a cone, then $K^*$ is also a cone on the same vertex (see, for example, \cite{FMS:StanleyReisner}). In light of \refP{DummyDictatorSimplex}, we thus topologically recover the following result which is familiar in simple game theory.

\begin{prop}
If $d$ is a dummy voter for a simple game, then $d$ is a dummy voter for the dual game as well.
\end{prop}

The following is a standard result in simple game theory, but we now restate it in terms of simplicial complexes.

\begin{prop}\label{P:ProperStrongSimplicial}\ 
\begin{itemize}
\item A simple game $(N,v)$ is strong if and only if $K_{\mathcal L}$ is a subcomplex of $K_{\mathcal L}^*$.
\item A simple game $(N,v)$ is proper if and only if $K_{\mathcal L}^*$ is a subcomplex of $K_{\mathcal L}$.
\end{itemize}
It follows that $(N,v)$ is constant sum if and only if the simplicial complex of losing coalitions is self-Alexander dual, i.e.~$K_{\mathcal L}=K_{\mathcal L}^*$.

\end{prop}

%

The last statement is of course parallel to \refP{SelfDualGame}. Since the proof is a simple unravelling of definitions, we leave it to the reader.

Recall from Definition \ref{D:Homomorphism} that a homomorphism of games $(N,v)$ and $(N',v')$ is a function $\phi\colon \mathcal P(N) \to \mathcal P(N')$ induced by a function $N\to N'$ that sends winning coalitions to winning coalition and losing coalitions to losing coalitions. In particular, a homomorphism induces a simplical map
$$
K_{\mathcal L}\longrightarrow K_{\mathcal L'}
$$
which is a homeomorphism if the homomorphism is an isomorphism.

Given a simplicial map as defined in \refD{SimplicialMap}, it is not as straightforward to see whether it induces a homomorphism of simple games. 

The first issue is that is that, given a simplicial complex and thinking of it as the complex of losing coalitions of some game, there is no canonical way to say how many players that game should have. This is because some of the players might themselves be winning coalitions, in which case they would not appear as vertices in $K_{\mathcal L}$. This can be defined away by only considering games with no passers (see \refD{KindsOfPlayers}) so that the vertex set of a simplicial complex corresponds to the players.

Beyond this, it is not easy to tell what kind of simplicial maps lift to maps of simple games. One can construct both injective and surjective simplicial maps representing losing coalitions that are not homomorphisms of games. 

For example, take $K=\{\{v_1\},\{v_2\}\}$ and $K'=\{\{v'_1\},\{v'_2\}, \{v'_1, v'_2\}\}$ with $s(v_i)=v'_i$. This is an injective simplicial map but since the set $\{\{v_1\},\{v_2\}\}$ is not in $K$, this means that this is a winning coalition. However, if this were a map of games, this winning coalition would map to $\{v'_1, v'_2\}$, which is a losing coalition.

On the other hand, take, for example, $K$ to be any simplicial complex which is not a simplex and $K'$ to be a single vertex gives a surjective (constant) simplicial map. But extending this to a homomorphism of games would mean that all the winning coalitions would also have to map to the single (losing) player, which is not allowed.

An interesting question, therefore, is to characterize those simplicial maps which extend to homomorphisms of simple games. Whatever this class is, it likely forms a subcategory of the category $\operatorname{SCom}$ of simplicial complexes. This would provide a category stucture for the collection of simple games which could yield interesting results and insights.

The following is the main new result of this paper, relating homology to symmetric weighted games. Recall that $\beta_i(X)$ denotes the rank of the $i$th homology group of $X$.

\begin{thm}\label{T:HomologySymmetric}
A simple game $(N, v)$ with $n$ players and the corresponding simplicial complex of losing coalitions $K_{\mathcal L}$ is symmetric with quota $q$ if and only if 
$$
\beta_i (K_{\mathcal L}) =
\begin{cases}
{n \choose q}, & i=q-2 ;\\
0, & \text{otherwise}.
\end{cases}
$$
\end{thm}

\begin{proof}
 For the forward direction, consider a symmetric simple game $(N, v)$ with $n$ players and quota $q$. Then any coalition consisting of $q$ or more players is winning and every coalition with fewer than $q$ players is losing. Thus every subset of $q-1$ or fewer players forms a $(q-2)$-simplex and its faces that is in $K_{\mathcal L}$, and this complex has no simplices of dimension higher than $q-2$. 

To construct a $(q-2)$-cycle, at least $q$ vertices and $q$ simplices are needed (a minimal 1-cycle is formed by three vertices and three 1-simplices forming a triangle, which is topologically a circle; a minimal 2-cycle is formed by four vertices and four 2-simplices forming a tetrahedron which is topologically a 2-sphere; etc.). If every subset of $q$ vertices forms a $(q-2)$-cycle, there are precisely $n \choose q$ of them, and that is the case in our situation. Since none of these cycles bound a higher-dimensional simplex, they are all distinct non-trivial homology classes, meaning that the rank of the homology group in dimension $q-2$ is precisely $n \choose q$.

 Since $K_{\mathcal L}$ contains no simplices of dimension greater than $q-2$, $\beta_{i} = 0$ for any $i > q - 2$. Similarly for $i<q-2$ since every cyle of such lower dimension can be filled in by simplices of one dimension higher as they are all present in the simplicial complex, meaning that every cycle is a boundary and hence homologically trivial. For example, if $\{v_1, ..., v_k\}$ forms a 1-cycle and $q\geq 4$, then there are ${k \choose 3}$ 2-simplices using those vertices that form a topological $(k-2)$-sphere on which this 1-cycle is homologous to the trivial one.
 
%
%
%

For the reverse direction, the assumption is that every possible $(q-2)$-simplex is present in $K_{\mathcal L}$. Because $K_{\mathcal L}$ is a simplicial complex, this also means that every possible subset of fewer than $q-2$ players is present (and is necessarily a face of a $(q-2)$-simplex; if it was not, the game would have more than $n$ players). In turn, this implies that every coalition of $q-1$ or fewer players is losing. Since none of these $(q-2)$-cycles bounds (by the trivial homology assumption), all coalitions of size $q$ or greater are winning, which means that this is a simple game with $n$ players and quota $q$. 
%
%
\end{proof}

We will make some remarks about potential generalizations of this result in Section \ref{S:FurtherWork}.

Finally, we remark on a ``complementary'' way to regard simple games as topological objects. Namely, if we think of each winning coalition as a simplex, then we can build a space from them, but now this would not be a simplicial complex because there would be no downward closure requirement as there is for losing coalitions. Instead, one would obtain a hypergraph with an upward closure condition. Alternatively, this can be regarded as the complement of $K_{\mathcal L}$ in $\Delta^{n-1}$, where as usual $|N|=n$.

The drawback of this point of view is that hypergraphs are topologically not as convenient to work with as simplicial complexes; see, for example, our remarks on computing their homology in \refS{Hypergraphs}. On the other hand, hypergraphs are more and more popular in modeling various processes where some objects interact in some ways (networks, signals, computational biology, etc.), so there is a well-developed theory of hypergraphs from graph-theoretic and other points of view that could potentially be brought to bear in the study of simple games.


\section{Restrictions of games and compatible coalitions}\label{S:Restrictions}


There is another way to bring simplicial complexes into the theory of simple games. Namely, we have thus far assumed that all coalitions are possible, or feasible. However, this assumption is clearly unreasonable in many real-world scenarios since subsets of players may not be compatible or aligned in some way (ideological, strategic, etc.) and would therefore never be in a coalition together. 

Restricting games to feasible coalitions has a long history in the literature. For example, Myerson \cite{M:GraphsCooperation} codifies restrictions through a graph which contains information about which players are in communication with one another (if they share an edge) and are hence compatible. Players sharing edges with multiple other players can act as intermediaries.

%
%

One can also take the opposite approach and model player incompatibilites by a graph. Now an edge between two players means that they can never be in a coalition together. Such games are sometimes called \emph{$I$-restricted} ($I$ is the relevant graph) and they have been studied extensively \cite{carreras:1991, bergatinos:1993, am:2015, yakuba:2008}. This work includes calculations of power indices of players in $I$ restricted games using generating functions.

%

Another way to restrict a game is to directly restrict its domain. Namely, if $S\in \mathcal P(N)$ is unfeasible, i.e.~consists of incompatible players, we remove $S$ from the domain of $v$. Since $\mathcal P(N)$ can be identified with $\Delta^{n-1}$ (ignoring the empty set which is a losing coalition by definition), this means that a certain face of  $\Delta^{n-1}$ is removed, i.e.~$v$ is restricted to $\Delta^{n-1}\setminus S$. All supersets of $S$ are also removed since any set containing incompatible players is unfeasible.

After all unfeasible coalitions are removed, what remains is a subcomplex $K\subset\Delta^{n-1}$ of feasible coalitions and the simple game is now a function
$$
v\colon K\longrightarrow \{0,1\}.
$$
This is an example of a cooperative game on a simplicial complex, which has been considered before \cite{Martino:GamesComplexes} (that paper looks at the more general setup of $v$ taking values in $\R$). Power indices for weighted voting systems restricted to simplicial complexes were studied in \cite{BSV:SimplicialPower}; that paper also has an extensive review of the literature considering restrictions of games to simplicial complex domains. Related work is \cite{AK:Conflict, MV:Politics} where political structures (rather than games) are modeled by simplicial complexes with simplices indicating feasible coalitions.

The above can accommodate graph $I$-restricted games mentioned earlier: Start with $\Delta^{n-1}$ and eliminate all 1-simplices that correspond to edges in $I$. Because simplicial complexes have downward closure, the elimination of these edges will necessitate the elimination of all the higher-dimension faces containing those pairs. The result is a simplicial complex of feasible coalitions.

The hypergraph of winning coalitions $K_{\mathcal W}\subset K$  consist of subsets of compatible players. If we also assume this about the simplicial complex of losing coalitions $K_{\mathcal L}\subset K$, we will refer to a game like that as \emph{doubly compatible}.

The formation of a feasible coalition does not imply, however, that its complement is feasible. But this might be a desirable condition since, for example, voters voting on one of two candidates (as is usually the case in the duopoly of the U.S. political system) or any group deciding on a binary choice naturally breaks up into two disjoint coalitions of compatible players.


%
%
%

%

We thus make the following 
\begin{defin}\label{def:observable coalition}
A feasible coalition is \emph{observable} if its complement is also feasible.
\end{defin}




Recall the notion of the Alexander dual from \refD{AlexDual}.

\begin{prop}\label{P:Observable} Given a simple game on a simplicial complex $K$ of feasible coalitions, the set of observable coalitions is given by $K - (K\cap K^*)$ where  $K^*$ is the Alexander dual of $K$.
\end{prop}

There is some abuse of notation here in that we mean to take the difference of simplicial complexes away from the singleton sets, namely away from vertices in $V$. Both $K$ and $K^*$ contain all elements of $V$, and we want to retain them in $K - (K\cap K^*)$.

\begin{proof}
By definition, $K^*$ contains precisely those simplices of $K$ whose complement does not appear in $K$. Removing these leaves those simplices of $K$ whose complement also appears in $K$.
\end{proof}

\begin{example}
Consider Example \ref{Ex:AlexanderDual}. We have
$$
K - (K\cap K^*) = \{\{v_1\}, \{v_2\}, \{v_3\}, \{v_4\}, \{v_1, v_4\}, \{v_2, v_3\}\}
$$
and the two coalitions are precisely complements of one another.
\end{example}

As a special case, if $K=\Delta^{n-1}$, namely all coalitions are feasible and all complements are present, then $K^*=\emptyset$ (again, by abuse of notation, this means $K^*$ consists only of vertices $V$) and $K-K^*=K$.
 
It is important to observe that $K - (K\cap K^*)$ is not necessarily a simplicial complex, just a hypergraph. Regardless, a game $v\colon K\to\{0,1\}$ defined on feasible coalitions $K$ can be restricted to a function
$$
v'=v\vert_{K - (K\cap K^*)}\colon K - (K\cap K^*)\longrightarrow\{0,1\}
$$
which is still a game since it continues to satisfy the upward closure on winning coalitions (since $v$ satisfies upward closure, so must its restriction $v'$; the only thing that changes is that some supersets that are in $K$, i.e.~in the domain of $v$, may no longer be in the domain of $v'$).

Finally, implicit in our discussion of restriction of simple games to simplicial complexes is the assumption that, if a coalition is compatible then all its subsets are also compatible; if this were not the case, we would not have a simplicial complex $K$ as the domain of $v$.  However, this assumption is somewhat limiting. For example, players $p_1$ and $p_2$ might only be in a coalition if there is another player $p_3$ present; think of $p_3$ as a moderator between $p_1$ and $p_2$. The coalition $\{p_1, p_2, p_3 \}$ is thus feasible, but $\{p_1, p_2\}$ may not be. 

If this setup is allowed, then the domain of $v$ is no longer a simplicial complex but a hypergraph. What we said so far applies to hypergraphs as well, including \refP{Observable} as the Alexander dual can be defined the same way for hypergraphs as for simplicial complexes (although that is not the traditional setting in which it is considered).


\section{Further work}\label{S:FurtherWork}


The purpose of this article is to begin the study of simple games from a topological viewpoint. We have only scratched the surface of the connections between the two fields and hope to inspire further work in creating the dictionary between them as well as reveal deeper structures in game theory using topological tools.

For example, understanding how different simplicial complex operations (barycentric subdivision, pushouts, posets associated to the skeletal structure, etc.) relate to game-theoretic transformations could be informative. Along these lines, if the simplicial complex of losing coalitions has interesting topology in the sense that it is, for example, non-orientable and is perhaps homeomorphic to the M\"obius strip or a projective space, what does that mean in terms of the game? Are games like that structurally different than those represented by, say, a torus or a sphere?

Another set of questions might involve the algebraic topology of simplicial complexes. If, say, the simplicial complex of losing coalitions is contractible, does that imply anything about the game? If the complex is acyclic, i.e.~has trivial homology, what can we say about the game? How does the operation of collapsing vertices (which is central in \cite{AK:Conflict, MV:Politics}) relate to operations on simple games?

Building on specific results in the paper, \refP{ProperStrongSimplicial} suggests that  investigating the combinatorial and topological properties of self-dual complexes could offer insights into the structure of constant-sum simple games and voting systems. Interpreting more general interactions of simplicial complex topology and Alexander duality might translate into new results about the interactions of sums and products with game duality. It would also clearly be desirable to generalize \refT{HomologySymmetric} to arbitrary weighted games. Bringing in the notion of \emph{weighted homology} \cite{BDLR:WeightedHomology, D:WeightedHomology, WRWX:Weighted(co)homology} might prove to be a useful way of making progress on this problem.

Section \ref{S:Restrictions} provides the beginning of the study of (doubly) compatible and observable coalitions. To explore this more, one would need to make a transition from simplicial complexes to hypergraphs (see \refP{Observable}). Although the two share many combinatorial properties (including the definition of the Alexander dual), the latter are not as nice topologically. In particular, the calculation of hypergraph homology can be tricky, and tools such as embedded homology mentioned in \refS{Hypergraphs} should be used. On the other hand, a switch to hypergraphs is desirable not only because of the content of  Section \ref{S:Restrictions}, but because the theory of simple games is based on the analysis of winning coalitions, and those  form a hypergraph. Related work, like the study of political structures \cite{AK:Conflict, MV:Politics}, also leads naturally to hypergraphs because a coalition might fall apart if a member leaves it, which translates into a face of a simplex not necessarily being present -- this is a violation of the main axiom defining simplicial complexes, but is allowed in hypergraphs.

Finally, we believe that there is much merit to employing the category-theoretic approach. Simplicial complexes form a category and its objects correspond to simple games, but, as discussed right before the statement of \refT{HomologySymmetric}, it is not immediately obvious what the morphisms on the simple game side ought to be. Once this is resolved, the rich structure of the simplicial complex category would almost certainly provide exciting new insights into game theory.


\bibliographystyle{alpha}


\begin{thebibliography}{WRWX23}

\bibitem[AK19]{AK:Conflict}
Joseph~M. Abdou and Hans Keiding.
\newblock A qualitative theory of conflict resolution and political compromise.
\newblock {\em Math. Social Sci.}, 98:15--25, 2019.

\bibitem[Bar93]{B:Impossibility}
Yuliy~M. Baryshnikov.
\newblock Unifying impossibility theorems: a topological approach.
\newblock {\em Adv. in Appl. Math.}, 14(4):404--415, 1993.

\bibitem[Bar97]{B:TopDiscSocChoice}
Yuliy~M. Baryshnikov.
\newblock Topological and discrete social choice: in a search of a theory.
\newblock volume~14, pages 199--209. 1997.
\newblock Topological social choice.

\bibitem[Bar23]{B:AroundArrow}
Yuliy Baryshnikov.
\newblock Around {A}rrow.
\newblock {\em Math. Intelligencer}, 45(3):224--231, 2023.

\bibitem[BDLR25]{BDLR:WeightedHomology}
Andrei~C. Bura, Neelav~S. Dutta, Thomas J.~X. Li, and Christian~M. Reidys.
\newblock A computational framework for weighted simplicial homology.
\newblock {\em Topology Appl.}, 360:Paper No. 109177, 15, 2025.

\bibitem[BLRW19]{BLRW}
Stephane Bressan, Jingyan Li, Shiquan Ren, and Jie Wu.
\newblock The embedded homology of hypergraphs and applications.
\newblock {\em Asian J. Math.}, 23(3):479--500, 2019.

\bibitem[BR24]{BR:GSThm}
Yuliy Baryshnikov and Joseph Root.
\newblock A topological proof of the {G}ibbard-{S}atterthwaite theorem.
\newblock {\em Econom. Lett.}, 234:Paper No. 111447, 4, 2024.

\bibitem[B{\v S}V]{BSV:SimplicialPower}
Anastasia Brooks, Franjo {\v S}ar{\v c}evi{\'c}, and Ismar Voli\'c.
\newblock Weighted simple games and the topology of simplicial complexes.
\newblock Submitted.

\bibitem[Car91]{carreras:1991}
Francesc Carreras.
\newblock Restriction of simple games.
\newblock {\em Mathematical Social Sciences}, 21(3):245--260, 1991.

\bibitem[Chi79]{C:FixedPoints}
Graciela Chichilnisky.
\newblock On fixed point theorems and social choice paradoxes.
\newblock {\em Econom. Lett.}, 3(4):347--351, 1979.

\bibitem[Chi80]{C:SocialChoiceTopology}
Graciela Chichilnisky.
\newblock Social choice and the topology of spaces of preferences.
\newblock {\em Adv. in Math.}, 37(2):165--176, 1980.

\bibitem[Chi82]{C:ParetoDictator}
Graciela Chichilnisky.
\newblock The topological equivalence of the {P}areto condition and the
  existence of a dictator.
\newblock {\em J. Math. Econom.}, 9(3):223--233, 1982.

\bibitem[Chi83]{C:SocialChoiceTopologyRecent}
G.~Chichilnisky.
\newblock Social choice and game theory: recent results with a topological
  approach.
\newblock In {\em Social choice and welfare ({C}aen, 1980)}, volume 145 of {\em
  Contrib. Econom. Anal.}, pages 79--102. North-Holland, Amsterdam, 1983.

\bibitem[Daw90]{D:WeightedHomology}
Robert J.~MacG. Dawson.
\newblock Homology of weighted simplicial complexes.
\newblock {\em Cahiers Topologie G\'eom. Diff\'erentielle Cat\'eg.},
  31(3):229--243, 1990.

\bibitem[FM98]{FM:VotingPower}
Dan~S. Felsenthal and Mosh\'{e} Machover.
\newblock {\em The measurement of voting power}.
\newblock Edward Elgar Publishing Limited, Cheltenham, 1998.
\newblock Theory and practice, problems and paradoxes.

\bibitem[FMS14]{FMS:StanleyReisner}
Christopher~A. Francisco, Jeffrey Mermin, and Jay Schweig.
\newblock A survey of {S}tanley-{R}eisner theory.
\newblock In {\em Connections between algebra, combinatorics, and geometry},
  volume~76 of {\em Springer Proc. Math. Stat.}, pages 209--234. Springer, New
  York, 2014.

\bibitem[FP09]{FPR09}
Davide~L. Ferrario and Renzo~A. Piccinini.
\newblock {\em Strutture simpliciali in topologia}, volume~50 of {\em Quad.
  Unione Mat. Ital.}
\newblock Bologna: Pitagora Editrice, 2009.

\bibitem[GBnGJ93]{bergatinos:1993}
Francesc~Carreras Gustavo Berganti\~nos and Ignacio Garc\'ia-Jurado.
\newblock Cooperation when some players are incompatible.
\newblock {\em Mathematical Methods of Operations Research}, 38:187--201, 1993.

\bibitem[JAM15]{am:2015}
M.G. Fiestras-Janeiro J.M. Alonso-Meijide, B. Casas-M\'endez.
\newblock Computing banzhaf-coleman and shapley-shubik power indices with
  incompatible players.
\newblock {\em Applied Mathematics and Computation}, 252(1):377--387, 2015.

\bibitem[LRRP]{LRR-P:GeneralArrow}
Isaac Lara, Sergio Rajsbaum, and Armajac Ravent{\'o}s-Pujol.
\newblock A generalization of arrow's impossibility theorem through
  combinatorial topology.

\bibitem[Mar21]{Martino:GamesComplexes}
Ivan Martino.
\newblock Cooperative games on simplicial complexes.
\newblock {\em Discrete Appl. Math.}, 288:246--256, 2021.

\bibitem[Mat03]{Mat03}
Ji{\v{r}}{\'{\i}} Matou{\v{s}}ek.
\newblock {\em Using the {Borsuk}-{Ulam} theorem. {Lectures} on topological
  methods in combinatorics and geometry. {Written} in cooperation with {Anders}
  {Bj{\"o}rner} and {G{\"u}nter} {M}. {Ziegler}}.
\newblock Universitext. Berlin: Springer, 2003.

\bibitem[MS05]{MS:CommAlg}
Ezra Miller and Bernd Sturmfels.
\newblock {\em Combinatorial commutative algebra}, volume 227 of {\em Graduate
  Texts in Mathematics}.
\newblock Springer-Verlag, New York, 2005.

\bibitem[MV21]{MV:Politics}
Andrea Mock and Ismar Voli{\'c}.
\newblock Political structures and the topology of simplicial complexes.
\newblock {\em Math. Soc. Sci.}, 114:39--57, 2021.

\bibitem[MVc21]{MV:PoliticsComplex}
Andrea {M}ock and Ismar {V}oli\' c.
\newblock Political structures and the topology of simplicial complexes.
\newblock {\em Mathematical Social Sciences}, 114:39--57, 2021.

\bibitem[Mye77]{M:GraphsCooperation}
Roger~B. Myerson.
\newblock Graphs and cooperation in games.
\newblock {\em Math. Oper. Res.}, 2(3):225--229, 1977.

\bibitem[RRP22]{RR-P:Arrow}
Sergio Rajsbaum and Armajac Ravent{\'o}s-Pujol.
\newblock A combinatorial topology approach to arrow's impossibility theorem.
\newblock \url{https://mpra.ub.uni-muenchen.de/113858/}, 2022.
\newblock MPRA Paper No. 113858.

\bibitem[RWZ22]{RWZ}
Shiquan Ren, Jie Wu, and Mengmeng Zhang.
\newblock The embedded homology of hypergraph pairs.
\newblock {\em J. Knot Theory Ramifications}, 31(14):Paper No. 2250103, 30,
  2022.

\bibitem[Sha62]{S:SimpleGames}
L.~S. Shapley.
\newblock Simple games: an outline of the descriptive theory.
\newblock {\em Behavioral Sci.}, 7:59--66, 1962.

\bibitem[Spa95]{Spa95}
Edwin~H. Spanier.
\newblock {\em Algebraic topology}.
\newblock Berlin: Springer-Verlag, 1995.

\bibitem[Tan06]{T:BrowerNash}
Yasuhito Tanaka.
\newblock On the equivalence of the {A}rrow impossibility theorem and the
  {B}rouwer fixed point theorem.
\newblock {\em Appl. Math. Comput.}, 172(2):1303--1314, 2006.

\bibitem[TZ92]{TZ:WeightedVoting}
Alan Taylor and William Zwicker.
\newblock A characterization of weighted voting.
\newblock {\em Proc. Amer. Math. Soc.}, 115(4):1089--1094, 1992.

\bibitem[TZ99]{TZ:SimpleGames}
Alan~D. Taylor and William~S. Zwicker.
\newblock {\em Simple games}.
\newblock Princeton University Press, Princeton, NJ, 1999.
\newblock Desirability relations, trading, pseudoweightings.

\bibitem[vNM44]{vNM:GameTheory}
John von Neumann and Oskar Morgenstern.
\newblock {\em Theory of {G}ames and {E}conomic {B}ehavior}.
\newblock Princeton University Press, Princeton, NJ, 1944.

\bibitem[VW]{VW:PoliticsHypergraph}
Ismar Voli\'c and Zixu Wang.
\newblock Political structures on hypergraphs.
\newblock To appear in Involve.

\bibitem[WRWX23]{WRWX:Weighted(co)homology}
Chengyuan Wu, Shiquan Ren, Jie Wu, and Kelin Xia.
\newblock Weighted (co)homology and weighted {L}aplacian.
\newblock {\em Houston J. Math.}, 49(2):397--429, 2023.

\bibitem[Yak08]{yakuba:2008}
V.~Yakuba.
\newblock Evaluation of {B}anzhaf index with restrictions on coalitions
  formation.
\newblock {\em Mathematical and Computer Modelling}, 48:1602--1610, 2008.

\bibitem[Yua]{Y:FixedPoints}
Allen Yuan.
\newblock Fixed point theorems and applications to game theory.
\newblock \url{https://math.uchicago.edu/~may/REU2017/REUPapers/Yuan.pdf}.

\end{thebibliography}

\end{document}